\documentclass[11pt]{article}
\usepackage{amsfonts,color}
\usepackage{amssymb}
\usepackage{amsmath}
\usepackage{amsthm}
\usepackage{natbib}
\usepackage{enumerate}
\usepackage{graphicx}
\usepackage{wrapfig}
\usepackage{rotating}
\usepackage{epsfig}
\usepackage{esvect}
\usepackage[left=1in, right=1in, top=1in, bottom=1in]{geometry}
\usepackage{mdwlist}
\usepackage{secdot}
\usepackage{comment}
\usepackage{mathrsfs}
\usepackage{mdwlist}
\usepackage{multirow}
\usepackage{lscape}
\usepackage{pdflscape}
\usepackage{array}
\usepackage[compact]{titlesec}
\usepackage{times}
\usepackage{tikz}
\usepackage{subfigure}
\usepackage[justification=centering]{caption}
\usepackage{setspace}
\usepackage{color}
\usepackage{authblk}
\usepackage{wrapfig}
\usepackage{hyperref}
\usepackage{makeidx}

\bibpunct{[}{]}{,}{n}{}{;}
\titlespacing*{\section}{0pt}{*2}{*0}
\titlespacing*{\subsection}{0pt}{*2}{*0}
\titlespacing*{\subsubsection}{0pt}{*1}{*0}

\newtheorem{theorem}{Theorem}

\newtheorem{proposition}{Proposition}

\newtheorem{definition}{Definition}

\newtheoremstyle{compacttheorem}
{3pt}
{3pt}
{}
{}
{}
{}
{}
{}

\title{Generalizing the Covering Path Problem on a Grid}
\author[1]{\rm Liwei Zeng}
\author[1]{\rm Karen Smilowitz}
\author[2]{\rm Sunil Chopra}
\affil[1]{Department of Industrial Engineering and Management Sciences}
\affil[2]{Kellogg School of Management}
\affil[1,2]{Northwestern University}
\date{}
\date{\vspace{-5ex}}
\begin{document}
\maketitle

\begin{abstract}
We study the covering path problem on a grid of $\mathbb{R}^{2}$. We generalize earlier results on a rectangular grid and prove that the covering path cost can be bounded by the area and perimeter of the grid. We provide (2+$\varepsilon$) and (1+$\varepsilon$)-approximations for the problem on a general grid and on a convex grid, respectively.
\end{abstract}


\section{Introduction} \label{sec:introduction}
The covering path problem (CPP) finds the cost-minimizing path connecting a subset of points in a network such that non-visited points are within a predetermined distance of a point from the subset. The CPP has been studied in the literature since \cite{current1981multiobjective} introduced the problem and proved its NP-hardness from a reduction of the traveling salesman problem (TSP). Existing work formulates the CPP as an integer linear program (\cite{current1989covering}); however, such a formulation can be challenging to implement in practice for large-scale instances. Recently, \cite{zeng2017covering} leverages the geometric structure of the coverage region and develops simple construction techniques with provable performance guarantee for the CPP on a rectangular grid with $l_{1}$ distance metric. Their work is motivated by school bus stop selection and routing in an urban setting with a grid-like road network. In this paper, we extend the results to more general grids to expand the applicability of the results.

As in \cite{zeng2017covering}, we consider a bi-objective CPP comprised of one cost term related to path length and one cost term related to stop count.

\begin{definition}[\textbf{CPP}]\label{def:CPP}
Given a coverage region $\mathcal{R}$ and a coverage radius $k>0$, the CPP finds a stop set $V\subseteq \mathcal{R}$ such that for every point in $\mathcal{R}$, there exists $v\in V$ at a distance at most $k$, and the minimum cost path $P_{\mathcal{R}}$ connecting all points in $V$. $P_{\mathcal{R}}$ is referred to as a covering path of $\mathcal{R}$. Let $L$ be the path length and let $T=|V|$ be the number of stops. Given $\alpha, \beta>0$, the cost of path $P_{\mathcal{R}}$ is defined as
	\begin{equation}\label{eqn:cost}
	Cost(P_{\mathcal{R}})=\alpha L+\beta T.
	\end{equation}
\end{definition}

Given a covering path $S_{1}\rightarrow S_{2}\rightarrow \cdots\rightarrow S_{T}$, each $S_{i}$ is called a \textit{stop}, $T$ is called the \textit{stop count} and $L=\sum_{i=1}^{T-1}||S_{i+1}-S_{i}||_{1}$ is called the \textit{path length}.

\cite{zeng2017covering} discusses three variants of the CPP on a rectangular grid, deriving results for all variants from fundamental results for the variant where the coverage region is a rectangular grid and stops can be located anywhere in the grid. Here, we extend those fundamental results to more general settings. Our approach also works for variants restricting the coverage region and stop locations to edges and integer points on the grid, respectively, with similar approximation results. We first define a general grid as follows.

\begin{definition}[Grid]\label{def:grid}
	A unit square in $\mathbb{R}^{2}$ is integral if its corners have integer coordinates. $G$ is a grid if it is the union of integral unit squares.
\end{definition}

We use $G$ to represent a general grid with area $A$ and perimeter $P$. For brevity, we use grid rather than general grid throughout the paper. Distance is measured using the $l_{1}$ norm on the grid. For coverage radius $k$ and point $(a,b)$, $D((a,b);k)=\{(x,y)~|~|x-a|+|y-b|\leq k\}$ denotes the diamond coverage region of $(a,b)$. We solve CPP where $\mathcal{R}$ is a general grid.

\begin{wrapfigure}{r}{0.5\textwidth}
	\centering
	\vspace{10pt}
	\includegraphics[width=0.5\textwidth]{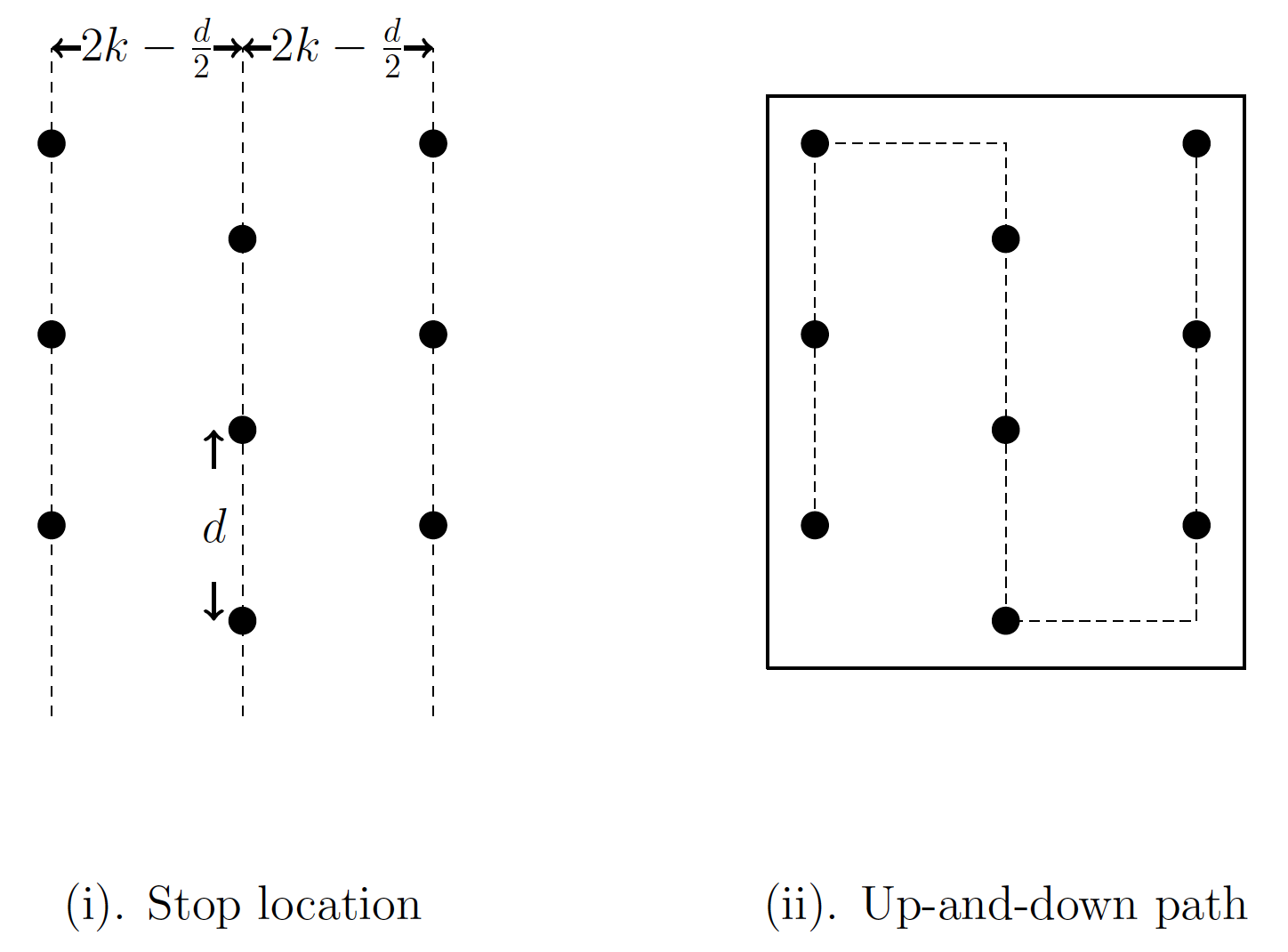}
	\vspace{-10pt}
	\caption{Covering path for a rectangular grid}
	\label{fig:tessellation}
\end{wrapfigure}

We summarize the path construction techniques for a rectangular grid in \cite{zeng2017covering}, which will be modified for a general grid. We use the term \textit{traversal} to represent a vertical line in $\mathbb{R}^{2}$, i.e., $x=constant$. Let $d\in (0,2k]$ be a distance parameter. On each traversal, stops are uniformly located such that the distance between consecutive stops is equal to $d$. To maintain coverage, the distance between consecutive traversals is equal to $(2k-\frac{d}{2})$ (Figure \ref{fig:tessellation} (i)). Stops are then connected in an up-and-down fashion as shown in Figure \ref{fig:tessellation} (ii). For a rectangular grid, \cite{zeng2017covering} shows that there exists $d\in (0,2k]$ such that the up-and-down path using parameter $d$ is near-optimal.

In Section \ref{sec:lower-bound}, we derive cost lower bound for the CPP on a grid using related results in \cite{zeng2017covering}. In Section \ref{sec:upper-bound}, we construct a covering path for a grid that provides $(2+\varepsilon)$-approximation solution for the CPP. We also study a special case where the coverage region is a convex grid and provide $(1+\varepsilon)$-approximation results.

\section{Cost Lower Bound for a General Grid}\label{sec:lower-bound}
\cite{zeng2017covering} presents the following relationship between stop count and path length.
\begin{theorem}[Trade-off Constraint \cite{zeng2017covering}]\label{thm:trade-off}
	Given a grid $G$ with area $A>2k^{2}$, any covering path of $G$ satisfies the trade-off constraint:
	\begin{equation}\label{eqn:trade-off}
	(T-1)f\Big(\frac{L}{T-1}\Big)\geq A-2k^{2},
	\end{equation}
	where $f(\cdot)$ is a function of the average distance between consecutive stops, $d_{avg}$, defined as:
	\begin{equation}
	f(d_{avg})=
	\begin{cases}\label{def:f}
	d_{avg}(2k-\frac{d_{avg}}{2})& \text{if}~~d_{avg}\in (0,2k],\\
	2k^{2}& \text{if}~~d_{avg}\in (2k,\infty).
	\end{cases}
	\end{equation}
\end{theorem}

$d_{avg}=\frac{L}{T-1}$ is the average distance between consecutive stops. The function $f(d_{avg})$ is an approximate measure of the unique coverage region for a single stop. To minimize $T$, from \eqref{eqn:trade-off}, we strive to maximize $f(d_{avg})$ (and therefore, $d_{avg}$). To minimize $L$, we tend to choose a smaller $d_{avg}$ so that stops are densely located along the covering path to minimize the number of traversals needed in an up-and-down path (see Figure \ref{fig:tessellation}(ii)). In this way, the average distance $d_{avg}$, together with the function $f(\cdot)$, controls the trade-off between the two objectives, $L$ and $T$.

Using \eqref{eqn:trade-off}, we provide a cost lower bound which is a linear function of the area of the grid.
\begin{proposition}\label{prop:cpp-opt}
	Given $\alpha,\beta>0$, a grid $G$ of area $A$, and coverage radius $k>0$. There exists $\sigma=\sigma(\alpha, \beta, k)>0$ such that 
	\[\alpha L+\beta T\geq \sigma (A-2k^{2}) \]
	for any covering path of $G$.
\end{proposition}

\begin{proof}
	From Theorem \ref{thm:trade-off}, the minimum value of $\alpha L+\beta T$ subject to the trade-off constraint \eqref{eqn:trade-off} is a lower bound of path cost. We show that this minimum value is at least $\sigma (A-2k^{2})$.
	
	Let $A_{0}=A-2k^{2}$ and let $T_{0}=T-1$. \eqref{eqn:trade-off} is equivalent to $T_{0}f(\frac{L}{T_{0}})\geq A_{0}$. Note that $\alpha L+\beta T=\alpha L+\beta T_{0}+\beta$, we consider the following optimization problem.
	
	\begin{align*} 
	&min~~\alpha L+\beta T_{0}+\beta \label{prb:1} \tag{OPT} \\
	&s.t.~~T_{0}f(\frac{L}{T_{0}})\geq A_{0}
	\end{align*}
	
	 To solve \eqref{prb:1}, since $f(\cdot)$ is an increasing function with upper bound $f(2k)$, it suffices to consider $(L,T_{0})$ pairs such that
	 
	 \[T_{0}f(\frac{L}{T_{0}})=A_{0}~~ \textrm{and}~~ \frac{L}{T_{0}}\leq 2k.\]
Expanding $f(\frac{L}{T_{0}})$, we have

\[L(2k-\frac{L}{2T_{0}})=A_{0}~~ \textrm{and}~~ T_{0}=\frac{L^{2}}{2(2kL-A_{0})}.\]
The objective function of \eqref{prb:1} can be written as a function $C(L)$ of $L$ where: 
	 
	 \[C(L)=\alpha L+\beta T_{0}+\beta=\alpha L+ \frac{\beta L^{2}}{2(2kL-A_{0})}+\beta.\]
	  To find the value of $L$ that minimizes $C(L)$, we set the first derivative to 0 to get
	  
	  \[C^{'}(L)=\alpha+\frac{\beta (kL^{2}-A_{0}L)}{(2kL-A_{0})^{2}}=0,\]
	  which implies
	  \[L=\frac{4\alpha k+\beta\pm\sqrt{(4\alpha k+\beta)\beta} }{2k(4\alpha k+\beta)} A_{0}.\]

	Note that $T_{0}=\frac{L^{2}}{2(2kL-A_{0})}\geq 0$ implies $L>\frac{A_{0}}{2k}$. Thus, $\gamma A_{0}$ is the only zero point of $C^{'}(L)$ in interval $(\frac{A_{0}}{2k}, \infty)$, where
	\[\gamma=\frac{4\alpha k+\beta+\sqrt{(4\alpha k+\beta)\beta} }{2k(4\alpha k+\beta)}.\]
	Since $C^{'}(L)$ is non-positive in $(\frac{A_{0}}{2k},\gamma A_{0})$, is equal to 0 at $\gamma A_{0}$, and is non-negative in $(\gamma A_{0},\infty)$, $\gamma A_{0}$ is the minimizer of $C(L)$. Therefore, the optimal solution to \eqref{prb:1} is:
	
	\[L^{*}=\gamma A_{0}, T^{*}_{0}=\frac{L^{2} }{4kL-2A_{0} }=\frac{\gamma^{2}}{4k\gamma-2}A_{0},\] 
	and the minimum path cost is $\alpha L^{*}+\beta T^{*}_{0}+\beta=(\alpha \gamma+\frac{\beta \gamma^{2}}{4k\gamma-2})A_{0}+\beta$. Let $\sigma=(\alpha \gamma+\frac{\beta \gamma^{2}}{4k\gamma-2})$. The path cost is at least $\sigma A_{0}+\beta\geq \sigma(A-2k^{2})$.
\end{proof}

Let $d_{avg}^{*}=\frac{L^{*}}{T^{*}-1}=\frac{L^{*}}{T_{0}^{*}}$. We rewrite $\sigma$ as a function of $d_{avg}^{*}$, which is used in Section \ref{sec:upper-bound} for cost upper bound estimation. Since $T_{0}^{*}f(\frac{L^{*}}{T_{0}^{*}})=A_{0}$, we have
\[T_{0}^{*}=\frac{A_0}{f(d_{avg}^{*})}, L^{*}=\frac{A_{0}d_{avg}^{*}}{f(d_{avg}^{*})}.\]
Given that $\alpha L^{*}+\beta T_{0}^{*}=\sigma A_{0}$, we have

\begin{equation}
\sigma=(\alpha L^{*}+\beta T_{0}^{*})/A_{0}=\Big(\alpha \cdot \frac{A_{0}d_{avg}^{*}}{f(d_{avg}^{*})}+\beta \cdot \frac{A_0}{f(d_{avg}^{*})}\Big)/A_{0}=\frac{\alpha d_{avg}^{*}+\beta}{f(d_{avg}^{*})}.
 \label{eqn:sigma}
\end{equation}

For a rectangular grid, \cite{zeng2017covering} minimizes $\alpha L+\beta T$ subject to \eqref{eqn:trade-off} to get an optimal $(L^{*},T^{*})$ and an optimal average distance $d_{avg}^{*}=\frac{L^{*}}{T^{*}-1}$. The up-and-down path in Figure \ref{fig:tessellation}(ii) using $d_{avg}^{*}$ is shown to be near-optimal for a rectangular grid. However, as we explore next, new path construction techniques are needed for a general grid.

\section{Cost Upper Bound for a General Grid}\label{sec:upper-bound}
We construct a covering path for a grid such that the path cost is bounded by the area and perimeter of the grid, which provides a (2+$\varepsilon$)-approximation for the CPP on a grid when the area grows significantly faster than the perimeter. The approximation ratio can be improved to (1+$\varepsilon$) when the grid is convex.

\begin{theorem}\label{thm:general-cost}
Given a grid $G$ of area $A$ and perimeter $P$, there exists a covering path $P_{G}$ and constants $c_{1},c_{2}$ such that
\[Cost(P_{G})\leq 2\sigma (A+c_{1}kP+c_{2}k^{2})\]
where $\sigma=\sigma(\alpha, \beta, k)$ is as defined in Proposition \ref{prop:cpp-opt}.
\end{theorem}

\cite{zeng2017covering} shows that the cost lower bound in Proposition \ref{prop:cpp-opt} is near-optimal when $G$ is a large rectangular grid (i.e. minimum cost $\approx \sigma A$). The results in Theorem \ref{thm:general-cost} look similar to that in \cite{zeng2017covering} when $A>>kP$, with an additional 2-approximation factor. As shown in Section \ref{subsec:cpp-convex}, the 2-approximation factor disappears when the grid is convex. Theorem \ref{thm:general-cost}, together with the lower bound in Proposition \ref{prop:cpp-opt}, gives a (2+$\varepsilon$)-approximation for grids where $A>>kP$. However, it differs significantly from \cite{zeng2017covering} when $A$ and $P$ are of the same order, which can occur in general.

\begin{wrapfigure}{r}{0.2\textwidth}
	\centering
	\vspace{-10pt}
	\includegraphics[width=0.2\textwidth]{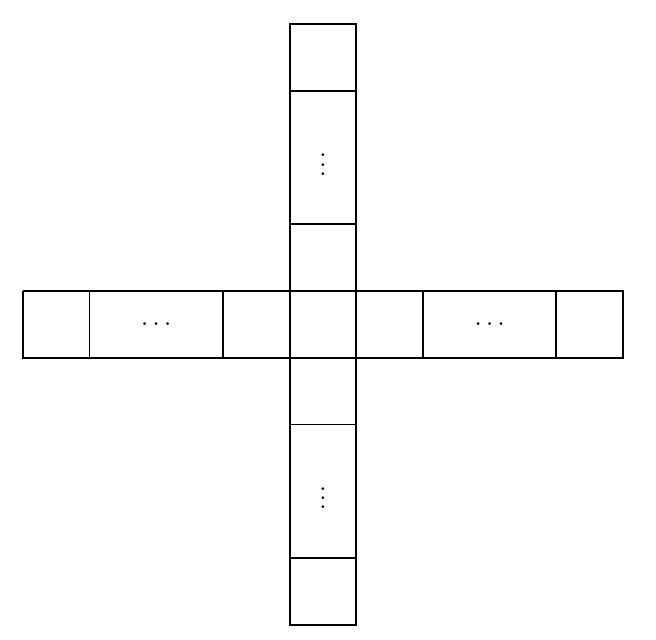}
	\vspace{-10pt}
	\caption{Cross}
	\label{fig:cross}
\end{wrapfigure}

Consider the grid in Figure \ref{fig:cross} with $n$ unit squares on each of the four arms. The grid has area $A=4n+1$ and perimeter $P=8n+4$. $A$ is no longer the dominating cost term since $A\approx 4n$ and $c_{1}kP\approx 8c_{1}kn$. Consider a special cost function $Cost=L$ with minimum value $L_{min}$. Results in \cite{zeng2017covering} show that 

\[L_{min}\geq \frac{A-2k^{2}}{2k}\approx\frac{2n}{k}=O\Big(\frac{n}{k}\Big)\]
while Theorem \ref{thm:general-cost} implies
\[L_{min}\leq \frac{A+c_{1}kP+c_{2}k^{2}}{2k}\approx 4c_{1}n=O(n).\]
It is not hard to see that the true value of $L_{min}$ is close to $8n$ (we have to go back and forth on each arm with cost $2n$), which means the bound in Theorem \ref{thm:general-cost} is tighter than that in \cite{zeng2017covering}.

Our proof strategy for Theorem \ref{thm:general-cost} is as follows: we solve the problem of minimizing $\alpha L+\beta T$ subject to the trade-off constraint \eqref{eqn:trade-off} to get an optimal $(L^{*},T^{*})$ and average distance $d_{avg}^{*}=\frac{L^{*}}{T^{*}-1}$. In \ref{subsec:stop-general} and \ref{subsec:path-general}, we use this $d_{avg}^{*}$ to construct a covering path of $G$. We provide bounds on the path length and stop count of the covering path and combine them to get a cost upper bound.

\subsection{Stop Selection}\label{subsec:stop-general}

\begin{wrapfigure}{r}{0.2\textwidth}
	\centering
	\includegraphics[width=0.2\textwidth]{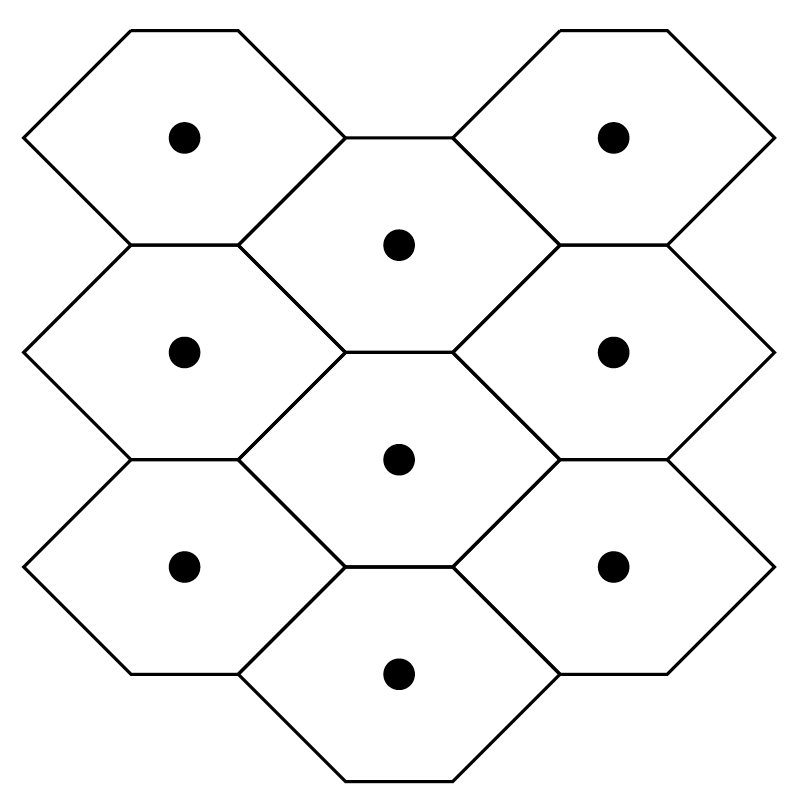}
	\caption{Tessellation of $\mathbb{R}^{2}$}
	\vspace{-10pt}
	\label{fig:tessellation1}
\end{wrapfigure}

Consider the stop locations in Figure \ref{fig:tessellation} (i) with parameter $d=d_{avg}^{*}$. We define a partition of $\mathbb{R}^{2}$ such that each point is assigned to the closest stop and that each stop covers a unique hexagon region of area $f(d_{avg}^{*})$ (Figure \ref{fig:tessellation1}). For a grid $G$, we select all hexagons from the tessellation that overlap with $G$. We construct a stop set such that each point in $G$ is covered by at least one point in the stop set.

Among the centers of these hexagons, let $C_{in}$ be the ones inside $G$ and let $C_{out}$ be the ones outside $G$. We first put all points in $C_{in}$ into the stop set. For each $X_{out}\in C_{out}$, \cite{zeng2017covering} shows that we can project $X_{out}$ to a point on the boundary of $G$ and maintain coverage when $G$ is a rectangular grid. Up to four additional points may be necessary to ensure coverage for a general grid.

\begin{wrapfigure}{r}{0.3\textwidth}
	\centering
	\vspace{15pt}
	\includegraphics[width=0.3\textwidth]{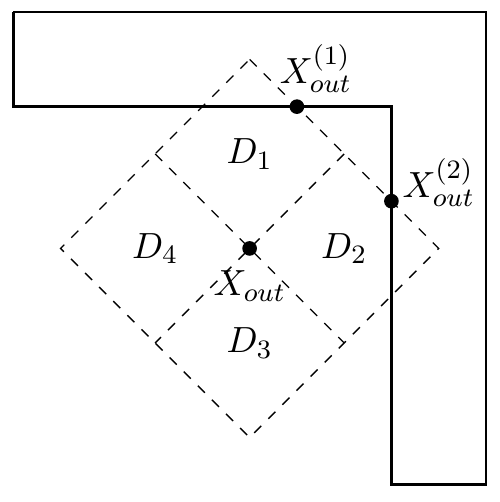}
	\vspace{10pt}
	\caption{Partition of $D(X_{out};k)$ and stop selection for $X_{out}$}
	\label{fig:Diamond partition}
\end{wrapfigure}

We partition the diamond coverage region of $X_{out}$ into four smaller diamonds $D_{1}, D_{2}, D_{3}$ and $D_{4}$ (see Figure \ref{fig:Diamond partition}, the coverage region can be viewed as part of the cross in Figure \ref{fig:cross}). For each $i\in \{1,2,3,4\}$ such that $D_{i}\cap G\neq \emptyset$, note that $X_{out}\in D_{i}$ and $X_{out}\notin G$. There exists a point $X_{out}^{(i)}\in D_{i}$ such that $X_{out}^{(i)}$ is on the boundary of $G$ (see $X_{out}^{(1)}$ and $X_{out}^{(2)}$ in Figure \ref{fig:Diamond partition}). Note that there can be multiple choices for each $X_{out}^{(i)}$. We put $X_{out}^{(i)}$ in the stop set for each $i$ such that $D_{i}\cap G\neq \emptyset$. In total, we select one stop for each point in $C_{in}$ and at most four stops for each point in $C_{out}$. Thus, the total stop count is at most $|C_{in}|+4|C_{out}|$.

\begin{proposition}\label{prop:cover-general}
Each point in $G$ is covered by at least one selected stop.
\end{proposition}

\begin{proof}
Recall that $D(X;k)$ is the diamond region covered by $X$ with coverage radius $k$. We have $G\subset \cup_{X\in C_{in}\cup C_{out}} D(X;k)$. All points in $C_{in}$ are in the stop set. Thus, it suffices to show that $D(X_{out};k)\cap G$ is covered by selected stops for any $X_{out}\in C_{out}$. Let $D_1, D_{2}, D_{3}$ and $D_{4}$ be the four diamonds from the partition of $D(X_{out};k)$. $\forall i\in \{1,2,3,4\}$, since $D_{i}$ is a diamond of radius $\frac{k}{2}$, each pair of points in $D_{i}$ are within distance at most $k$ of each other. $\forall i\in \{1,2,3,4\}$, if $D_{i}\cap G\neq \emptyset$, the selected stop $X_{out}^{(i)}$ covers all points in $D_{i}$, therefore, covers all points in $D_{i}\cap G$.
\end{proof}

The following proposition provides an upper bound on the number of selected stops.

\begin{proposition}\label{prop:stop-count-general}
$|C_{in}|+|C_{out}|\leq \frac{A+4kP+8k^{2}}{f(d_{avg}^{*})}$, $|C_{out}|\leq \frac{4kP+8k^{2}}{f(d_{avg}^{*})}$.
\end{proposition}

\begin{proof}
Let $\widetilde{G}$ be region covered by points in $C_{in}\cup C_{out}$. We have $G\subseteq \widetilde{G}$. Since any point in $C_{in}\cup C_{out}$ is either inside $G$ (those in $C_{in}$) or within distance $k$ of the boundary of $G$ (those in $C_{out}$), each point in $\widetilde{G}/G$ is within distance $k+k=2k$ of some point on the boundary of $G$ (the first $k$ bounds the distance from this point to a stop and the second $k$ bounds the distance from the stop to the boundary of $G$). Similar to the proof of Lemma 10 in \cite{carlsson2014continuous}, the area of region covered by the boundary of $G$ with coverage radius $2k$ is at most $4kP+8k^{2}$, i.e. $|\widetilde{G}/G|\leq 4kP+8k^{2}$. Therefore, $|\widetilde{G}|\leq A+4kP+8k^{2}$. Since each point in $C_{in}\cup C_{out}$ covers a unique hexagon of area $f(d_{avg}^{*})$ in the tessellation of $\mathbb{R}^{2}$, the total number of points in $C_{in}\cup C_{out}$ is at most $\frac{A+4kP+8k^{2}}{f(d_{avg}^{*})}$. Similarly, since the region covered by points in $C_{out}$ is a subset of region covered by the boundary of $G$  with coverage radius $2k$, we have $|C_{out}|\leq \frac{4kP+8k^{2}}{f(d_{avg}^{*})}$.
\end{proof}

From Proposition \ref{prop:stop-count-general}, the number of selected stops $T$ is at most
\begin{equation}\label{eqn:stop-count}
T\leq |C_{in}|+4|C_{out}|=|C_{in}|+|C_{out}|+3|C_{out}|\leq \frac{A+16kP+32k^{2}}{f(d_{avg}^{*})}.
\end{equation}

\subsection{Path Generation}\label{subsec:path-general}
We construct a spanning tree of the selected stops and use the spanning tree length to give an upper bound on path length.

\begin{proposition}\label{prop:path-length-general}
There exists a spanning tree of the selected stop set such that the tree length is at most $|C_{in}|d_{avg}^{*}+P$.
\end{proposition}

\begin{proof}
First, we connect all pairs of stops in $C_{in}$ that are consecutive (with distance $d_{avg}^{*}$) on the same traversal. This separates $C_{in}$ into several connected components. Each connected component consists of consecutive stops can be connected to the boundary of $G$ with an edge of length at most $d_{avg}^{*}$. The total edge lengths connecting stops in one connected component to the boundary of $G$ is thus at most $d_{avg}^{*}$ times the number of stops in the connected component. Summing over all connected components, the cost to connect all points in $C_{in}$ to the boundary of $G$ is at most $|C_{in}|d_{avg}^{*}$.

Note that stops that are not in $C_{in}$ are on the boundary of $G$; thus, we connect all selected stops with length at most $|C_{in}|d_{avg}^{*}+P$ with an additional cycle around the perimeter for stops on the boundary. This proves that the minimum spanning tree (MST) length of the stop set is at most $|C_{in}|d_{avg}^{*}+P$.
\end{proof}

Recall that the TSP length is at most twice as the MST length (\cite{lawler1985traveling}). The minimum path length $L$ connecting all selected stops, which is less than the TSP length of the selected stop set, is thus bounded by
\begin{equation}\label{eqn:path-length}
L\leq 2(|C_{in}|d_{avg}^{*}+P)\leq \frac{2d_{avg}^{*}(A+4kP+8k^{2})}{f(d_{avg}^{*})}+2P\leq \frac{2d_{avg}^{*}(A+6kP+8k^{2})}{f(d_{avg}^{*})}.
\end{equation}

\subsection{Proof of Theorem \ref{thm:general-cost}}\label{subsec:proof-thm2}
Using \eqref{eqn:stop-count} and \eqref{eqn:path-length}, the total path cost is at most
\[\alpha \cdot \frac{2d_{avg}^{*}(A+6kP+8k^{2})}{f(d_{avg}^{*})}+\beta\cdot \frac{A+16kP+32k^{2}}{f(d_{avg}^{*})}\leq (A+16kP+32k^{2})\cdot \frac{2(\alpha f(d_{avg}^{*})+\beta)}{f(d_{avg}^{*})}. \]

Using \eqref{eqn:sigma}, the right hand side equals to $2\sigma (A+16kP+32k^{2})$. Therefore, Theorem \ref{thm:general-cost} holds with $c_{1}=16, c_{2}=32$.

\subsection{Improved Upper Bound for a Convex Grid}\label{subsec:cpp-convex}
As an extension, we strengthen the cost upper bound in Theorem \ref{thm:general-cost} when $G$ is a convex grid. Informally, a grid is convex if it avoids $U$-shape subregions.
\begin{definition}[Convex Grid]\label{def:convex-grid}
	$G\subset \mathbb{R}^{2}$ is an orthogonal convex set if it is contiguous and for every line parallel to the $x-$axis or $y-$axis, the intersection of $G$ and the line is a point, a line segment or an empty set. A grid is convex if it is an orthogonal convex set.
\end{definition}

Note that the (2+$\varepsilon$)-approximation factor comes from path length estimation: the ratio between TSP length and MST length. When $G$ is a convex grid, we can simply connect the connected components defined in \ref{subsec:path-general} in an up-and-down fashion (with the path structure for a rectangular grid) so that the total path length is at most $|C_{in}|d_{avg}^{*}+2P$. Together with the upper bound on stop count in \ref{subsec:stop-general}, the path cost is at most $\sigma(A+16kP+32k^{2})$ for a convex grid, providing a (1+$\varepsilon$)-approximation for the CPP when $A>>kP$.

\section{Conclusion and Discussion}\label{sec:conclusion}
In this paper, we study the covering path problem on a grid of $\mathbb{R}^{2}$. We derive a cost lower bound based on previous results for the CPP on a rectangular grid. We then complement with a cost upper bound which is a function of area and perimeter of the grid. These results together, provide (2+$\varepsilon$)-approximation for CPP on a large grid and (1+$\varepsilon$)-approximation on a large convex grid.

Our results can be applied to estimate transportation costs for school zones with arbitrary shapes. For a given school zone, Proposition \ref{prop:cpp-opt} and Theorem \ref{thm:general-cost} give lower and upper bounds for a single path's travel cost. Prior research has shown that most school zones are ``well-behaved" when the objective is to maximize compactness of attendance areas (\cite{bouzarth2018assigning},\cite{carlsson2016shadow},\cite{caro2004school},\cite{lemberg2000school}), allowing our results for a convex grid to be applied to obtain (1+$\varepsilon$)-approximation.

Our methodology can incorporate other realistic considerations and yield similar approximation results. For school within a walking zone where students within a predetermined distance walk to school, the coverage region is a grid with an interior hole. Our bounds hold in this case where a grid's perimeter is the summation of its exterior and interior boundary lengths. To incorporate bus capacity, a single path can be partitioned into a series of paths according to bus capacity so that the total path cost remains the same, plus an additional detour term. These extensions are being used in ongoing work on school district planning.

\textbf{Acknowledgements}

\noindent This project is supported by the National Science Foundation (CMMI-1727744).

\footnotesize
\singlespacing
\bibliographystyle{abbrvnat}
\bibliography{Literature}

\end{document}